\documentclass{article}

\usepackage[english]{babel}
\usepackage[letterpaper,top=2cm,bottom=2cm,left=3cm,right=3cm,marginparwidth=1.75cm]{geometry}
\usepackage{amsmath}
\usepackage{amsthm}
\usepackage{natbib}
\usepackage{graphicx}
\usepackage{verbatim}
\usepackage{dsfont}
\usepackage{float}
\usepackage[colorlinks=true, allcolors=blue]{hyperref}
\usepackage[ruled]{algorithm2e}
\usepackage{appendix}
\usepackage{bm}
\usepackage{enumerate}
\usepackage{lipsum}

\newcommand{\xvec}{\bm{x}}
\newcommand{\yvec}{\bm{y}}
\newcommand{\zvec}{\bm{z}}
\newcommand{\uvec}{\bm{u}}
\newcommand{\nuvec}{\bm{\nu}}
\newcommand{\fvec}{\bm{f}}
\newcommand{\deltavec}{\bm{\delta}}
\newcommand{\kappavec}{\bm{\kappa}}
\newcommand{\alphavec}{\bm{\alpha}}
\newcommand{\betavec}{\bm{\beta}}
\newcommand{\nullvec}{\bm{0}}
\newcommand{\epsilonvec}{\bm{\epsilon}}
\newcommand{\varphivec}{\bm{\varphi}}

\newcommand{\Amat}{\bm{A}}

\newcommand{\Dmat}{\bm{D}}
\newcommand{\Emat}{\bm{E}}

\newcommand{\Hmat}{\bm{H}}
\newcommand{\Imat}{\bm{I}}

\newcommand{\Kmat}{\bm{K}}

\newcommand{\Pmat}{\bm{P}}
\newcommand{\Qmat}{\bm{Q}}
\newcommand{\Rmat}{\bm{R}}
\newcommand{\Smat}{\bm{S}}
\newcommand{\Tmat}{\bm{T}}
\newcommand{\Umat}{\bm{U}}

\newcommand{\Wmat}{\bm{W}}
\newcommand{\Xmat}{\bm{X}}

\newcommand{\Zmat}{\bm{Z}}

\newcommand{\Gammamat}{\bm{\Gamma}}

\newcommand{\Lambdamat}{\bm{\Lambda}}
\newcommand{\Phimat}{\bm{\Phi}}

\newcommand{\nullmat}{\bm{0}}

\newtheorem{theorem}{Theorem}
\newtheorem{lemma}{Lemma}
\newtheorem{definition}{Definition}
\newtheorem{remark}{Remark}
\newtheorem{proposition}{Proposition}

\title{Truly optimal low rank thin plate spline smoothing using a truncated Demmler-Reinsch basis}
\author{Paul Bach, Berlin, Germany}

\begin{document}
\maketitle
\begin{abstract} Thin plate splines are highly attractive smoothers. However, they have cubic computational cost, which severely limits their use in practice. As a remedy, Wood (2003) suggested thin plate regression splines (TPRS), which provide a low rank approximation. The key step of the TPRS approximation is a truncated eigendecomposition of the radial basis function (RBF) design matrix. However, as Wood (2003) writes, the optimality of the TPRS approximation is a slightly weak one. This is because the RBF coefficients are subject to orthogonality constraints and the TPRS approximation is only optimal if these constraints are ignored. To address this shortcoming, we suggest a slightly different low rank approximation. The suggested approximation is based on a truncated Demmler-Reinsch basis (TDRB), which provides a best low rank approximation of the smoother matrix in terms of Frobenius and spectral norm. We prove that the TDRB smoother achieves the optimal rate of convergence and suggest an efficient algorithm for its construction. This algorithm is based on a truncated Karhunen-Lo\`eve (KL) expansion of the equivalent Bayesian smoothness prior and it has the same computational cost as required for TPRS. We demonstrate the applicabilty of our approach through simulations and a real data example. We find that the performance is very similar to that of TPRS but the suggested approach has some advantages. 
\end{abstract}


\section{Low rank thin plate spline smoothing}
Let $(\xvec_i,y_i),\ i=1,\dots,n,$ be some data points with distinct inputs $\xvec_i\in\mathds{R}^d$ of dimension $d\geq 1$ and real-valued outputs $y_i\in\mathds{R}$. Following \citet{duchon_splines_1977,wahba_spline_1990} the thin plate spline of integer order $m>d/2$ is defined as a minimizer of the penalized sum of squares
\begin{align}\label{PenalizedSumSquares}
   \widehat{f}_\lambda \in  \underset{f\in \text{BL}^m(\mathds{R}^d)}{\text{argmin}}\left\lbrace 1/n \sum_{i=1}^n (y_i-f(\xvec_i))^2 + \lambda J_{m}(f)\right\rbrace,
\end{align}
where $J_{m}(f)=\sum_{|\nuvec|=m} \frac{m!}{\nuvec!}\int_{\mathds{R}^d} \{D^{\nuvec} f(\xvec)\}^2 d\xvec$
is a roughness penalty with multi-index $\nuvec=(\nu_1,\dots,\nu_d)$. Moreover, $\lambda>0$ is a smoothing parameter and $\text{BL}^m(\mathds{R}^d)$ is the Beppo Levi space of order $m$. This is a Sobolev-type function space of continuous functions whose $m$-th order partial derivatives (in the distributional sense) are square integrable, i.e.,
\begin{align}\label{BeppoLevi}
    \text{BL}^m(\mathds{R}^d)=\{f\in C(\mathds{R}^d):  D^{\nuvec} f \in L^2(\mathds{R}^d),\ |\nuvec|=m\};
\end{align}
see \cite{meinguet1979multivariate} for further details. The null space of the penalty $J_m(\cdot)$ is the space of $d$-variate polynomials of degree $m-1$ or less, which we denote by $\mathcal{P}_{m-1}$. The dimension of $\mathcal{P}_{m-1}$ is $M=\binom{d+m-1}{m-1}$ and it is well known that there exists a unique minimizer $\widehat{f}_\lambda$ if the $\xvec_i$ are such that least squares regression onto $\mathcal{P}_{m-1}$ has a unique solution~\citep{wahba_spline_1990}. This condition is easily satisfied and will be assumed throughout. For $d=2$, for instance, the $\xvec_i\in\mathds{R}^2$ must not lie on a line.
 
To compute the thin plate spline, RBFs are commonly used \citep[see, e.g.,][]{green_nonparametric_1994}. To this end, let 
for $\xvec\in \mathds{R}^d$ and $j=1,\dots,n,$ 
\begin{align*}
    e_j(\xvec)=\begin{cases}
     \dfrac{(-1)^{m+1+d/2}}{2^{2m-1}\pi^{d/2}(m-1)!(m-d/2)!}\|\xvec-\xvec_j\|^{2m-d}\log\|\xvec-\xvec_j\|, & \text{if $d$ is even},\\\\
    \dfrac{\Gamma(d/2-m)}{2^{2m}\pi^{d/2}(m-1)!}\|\xvec-\xvec_j\|^{2m-d}, & \text{if $d$ is odd},
  \end{cases}
\end{align*}
where $\|\cdot\|$ is the Euclidean norm. 
Let further $\{t_j\}_{j=1}^M$ be a basis of $\mathcal{P}_{m-1}$ such as the monomials. With this, the thin plate spline can be represented in the form
\begin{align*}
    \widehat{f}_\lambda (\xvec) =\sum_{j=1}^M t_j(\xvec)\widehat{\alpha}_j + \sum_{j=1}^n e_j(\xvec)\widehat{\delta}_j,\ \xvec\in \mathds{R}^d,
\end{align*}
 with coefficient vectors $\widehat{\alphavec}$ and $\widehat{\deltavec}$ obtained as solutions to the constraint optimization problem
\begin{align*}
   \underset{(\alphavec,\deltavec)\in\mathds{R}^{M+n}}{\text{minimize}}\quad \left\lbrace 1/n\ \|\yvec-\Tmat\alphavec-\Emat\deltavec\|^2+\lambda \deltavec^t \Emat \deltavec\right\rbrace\  \quad \text{subject to} \quad \Tmat^t\deltavec=\nullvec.
\end{align*}
Thereby, $\yvec=(y_1,\dots,y_n)^t$ is the response vector, $\Emat=\{e_j(\xvec_i)\}$ is the $n\times n$ RBF design matrix and 
$\Tmat=\{t_j(\xvec_i)\}$ is the $n\times M$ polynomial design matrix.
Two important special cases are the univariate and the bivariate thin plate spline with second order penalization. In the former case $\Tmat$ has rows $(1,x_i)$ and $\Emat$ has entries $1/12\ |x_i-x_j|^3$. In the latter case, $\Tmat$ has rows $(1,x_{i1},x_{i2})$ and $\Emat$ has entries $1/(8\pi)\ \|\xvec_i-\xvec_j\|^2\log \|\xvec_i-\xvec_j\|,\ i,j=1,\dots,n$. 

Thin plate splines are highly attractive smoothers: They work well for different input dimensions and achieve the optimal rate of convergence \citep{utreras_convergence_1988}. From a practical perspective, one does not need to worry about an adequate number of knots and their position.
However, thin plate splines have cubic computational cost $O(n^3)$, which severely limits their use. 

A simple approach to reduce the computational cost is to use a random subset of the inputs $\xvec_i$ as knots \citep{kim2004smoothing}.
Although one can still achieve the optimal rate of convergence, this approach yields suboptimal performance in practice. 
Seeking for an optimal low rank approximation, \citet{wood_thin_2003} suggested thin plate regression splines (TPRS). The key idea is to restrict the RBF coefficients $\deltavec$ to a low dimensional subspace of $\mathds{R}^n$. More specifically,
one first computes a truncated eigendecomposition of the RBF design matrix $\Emat\approx \Umat_k\Lambdamat_k\Umat_k^t$, where $k$ is the number of largest magnitude eigenvalues. Then one restricts $\deltavec$ to $\text{span}(\Umat_k)\subseteq \mathds{R}^n$ before realizing the orthogonality constraints $\Tmat^t\deltavec=\nullvec$. Typically $k$ is chosen much smaller than $n$, so that the restriction to $\text{span}(\Umat_k)$ reduces the computational cost significantly. TPRS are the default smoother in Wood's popular \texttt{R} package \texttt{mgcv} and they work well in practice. However, the truncated eigendecomposition only provides an optimal approximation of $\Emat$ if one ignores the constraints $\Tmat^t\deltavec=\nullvec$. Therefore, Wood speaks of \emph{weak optimality}; see \citet[][end of Section 2.1 and Section 5]{wood_thin_2003} as well as \citet[p. 259]{wood_generalized_2017}.

The main purpose of this paper is to introduce a slightly different low rank approximation of thin plate splines. The suggested approximation is closely related to TPRS but can be regarded as a \emph{truly optimal} as the constraints $\Tmat^t\deltavec=\nullvec$ are incorporated. In summary, the key contributions of this paper are as follows:

\begin{itemize}
    \item We introduce a low rank approximation of thin plate splines, which we refer to as TDRB smoother. We prove that the TDRB smoother provides an optimal approximation and that it achieves the optimal rate of convergence.
    \item We devise a computationally efficient algorithm of cost $O(n^2k)$ to construct the TDRB smoother. This is exactly the same computational cost as required to set up TPRS \citep{wood_thin_2003}. 
    \item We provide empirical evidence demonstrating the competitive performance for multivariate smoothing and for generalized additive models (GAMs).
\end{itemize}

The remainder of this paper is structured as follows: Section~\ref{sec:tDRB smoother} introduces the TDRB smoother, proves its optimality and that it achieves the optimal rate of convergence. Section~\ref{sec:Construction} suggests a computationally efficient algorithm to construct the TDRB smoother,
while Section~\ref{sec:EmpiricalEvidence} provides empirical evidence. Section~\ref{sec:Discussion} closes with a discussion. Appendix~\ref{sec:Appendix} contains proofs and explains how to best handle replicated inputs. 


\section{The TDRB smoother}\label{sec:tDRB smoother}
This section introduces the TDRB smoother. First, we need to introduce additional notation: We use $(\cdot,\cdot)_n$ for the empirical or design (semi-) inner product, i.e., $(f,g)_n=1/n\sum_{i=1}^n f(\xvec_i)g(\xvec_i)$ and $\|\cdot\|_n$ for the associated (semi-) norm, i.e., $\|f\|^2_n=1/n\sum_{i=1}^n f(\xvec_i)^2$. In addition, we use $\mathcal{B}_m$ to denote the bilinear form that is associated with the roughness penalty, i.e.,
\begin{align}\label{RoughnessBilinear}
\mathcal{B}_m(f,g)=\sum_{|\nuvec|=m} \frac{m!}{\nuvec!}\int_{\mathds{R}^d} \{D^{\nuvec} f(\xvec)\}\{D^{\nuvec} g(\xvec)\} d\xvec,\quad f,g\in \text{BL}^m(\mathds{R}^d).    
\end{align}
Finally, we use $\mathcal{S}_n=\{f=\sum_{j=1}^M t_j\alpha_j + \sum_{j=1}^n e_j\delta_j : \Tmat^t\deltavec=\nullvec\}\subset \text{BL}^m(\mathds{R}^d)$ to denote the $n$-dimensional space of natural thin plate splines, which contains the thin plate spline estimate $\widehat{f}_\lambda$. 

Following the ideas of \cite{demmler_oscillation_1975, speckman_spline_1985} for smoothing splines and \citet[][Appendix A]{bach2026posterior} for Bayesian penalized B-splines, 
we define a Demmler-Reinsch basis as a basis of $\mathcal{S}_n$ that is orthonormal with respect to the design inner product and orthogonal with to respect $\mathcal{B}_{m}$. Definition~\ref{Def1} makes this precise.
\begin{definition}\label{Def1}
    A basis $\{\varphi_j\}_{j=1}^n$ of the natural thin plate splines $\mathcal{S}_n$ is a Demmler-Reinsch basis if  
\begin{align*}
    (\varphi_j,\varphi_l)_n=\delta_{j,l}\quad \text{and}\quad \mathcal{B}_{m}(\varphi_j,\varphi_l)=\gamma_j\delta_{j,l},
\end{align*}
for $j,l=1,\dots,n$, where $\delta_{j,l}$ is Kronecker's delta and $0=\gamma_1=\dots=\gamma_M < \gamma_{M+1}\leq \dots \leq \gamma_n$. 
 \end{definition}

To reduce the computational burden of smoothing, we suggest to restrict the minimization of the penalized sum of squares~\eqref{PenalizedSumSquares} to the subspace $\text{span}(\{\varphi_j\}_{j=1}^k)\subseteq \mathcal{S}_n$ that is spanned by the first $k$ basis functions of a Demmler-Reinsch basis with $k\in\{M+1,\dots,n\}$. We denote the resulting truncated Demmler-Reinsch basis (TDRB) smoother by $\widehat{f}_\lambda^k$. In which sense is the TDRB smoother an optimal approximation of the thin plate spline?
Proposition~\ref{Prop:FrequentistOptimality} below provides an answer.

\begin{proposition}~\label{Prop:FrequentistOptimality}
The smoother matrix $\Smat_\lambda^k$ of the TDRB smoother $\widehat{f}_\lambda^k$ is a best rank $k$ approximation of the smoother matrix $\Smat_\lambda$ of the thin plate spline $\widehat{f}_\lambda$ in terms of Frobenius norm and spectral norm.    
\end{proposition}

\begin{proof}
    See Appendix~\ref{sec:ProofTheo1}.
\end{proof}
Next, we investigate the asymptotic behavior of the TDRB smoother. We make the following assumptions: 
\begin{itemize}
    \item[A1)] The data is generated from the nonparametric regression model $y_i=f_0(\xvec_i)+\epsilon_i,\ i=1,\dots,n,$ where the residuals $\epsilon_i$ are uncorrelated with mean zero and variance $\sigma^2>0$.  
    \item[A2)] The true function $f_0$ is Sobolev-regular of integer order $m>d/2$ on a bounded domain $\Omega\subset\mathds{R}^d$, which has a Lipschitz boundary and satisfies a uniform cone condition.  
    \item[A3)] The $\xvec_i,\ i=1,\dots,n,$ are distinct and such that least squares regression onto $\mathcal{P}_{m-1}$ has a unique solution. In addition, the $\xvec_i,\ i=1,\dots,n,$ are quasi-uniform on $\Omega$, i.e., there exists a constant $B>0$ such that
   ${\sup_{\xvec\in \Omega} \inf_{i=1,\dots,n} \|\xvec-\xvec_i\|}/{\min_{i\neq j} \|\xvec_i-\xvec_j\|}\leq B.$
   \item[A4)] The smoothing parameter decays at rate $\lambda_n \asymp n^{-2m/(2m+d)}$. 
   \item[A5)] The number of basis functions grows at rate $k_n \asymp n^\delta$ with $\delta\in[d/(2m+d),n]$.
\end{itemize}
Assumptions A1)--A4) have been used by \cite{utreras_convergence_1988} to establish asymptotic results for thin plate splines, A5) is a new assumption for the TDRB smoother.

\begin{theorem}~\label{TheoAsymptotics}
 Under assumptions A1)--A5), the TDRB smoother achieves the optimal rate of convergence, i.e., $\mathds{E}_0 \|\widehat{f}_{\lambda_n}^{k_n}-f_0\|_n^2=O(n^{-2m/(2m+d)})$.  
\end{theorem}

\begin{proof}
    See Appendix~\ref{sec:ProofAsymptotics}.
\end{proof}

Theorem~\ref{TheoAsymptotics} shows that with adequate penalization the TDRB smoother achieves the optimal rate of convergence. Thereby, the number of basis functions $k_n$ can grow at a much slower rate than the number of observations $n$. This allows for computational savings without loss of performance.
However, for the approach to be actually useful in practice, we need a computationally efficient algorithm to construct the TDRB smoother. This is the topic of the following Section~\ref{sec:Construction}.


\section{Computationally efficient Bayesian construction}\label{sec:Construction}
This section introduces a computationally efficient construction of the TDRB smoother. The construction is motivated by the reparametrization of \citet[][Section 2.1]{scheipl_spike-and-slab_2012} for Bayesian penalized B-splines and relies on a discrete KL expansion of the equivalent Bayesian smoothness prior.

The thin plate spline $\widehat{f}_\lambda$ can be obtained as posterior mode in the Gaussian response model
$y_i=f(\xvec_i)+\epsilon_i,\ \epsilon_i\overset{iid}{\sim} N(0,\sigma^2),\ i=1,\dots,n,$
when using the prior
\begin{align}\label{BayesianSmoothnessPrior}
f=f_{poly}+f_{gp}=\sum_{j=1}^M t_j\alpha_j + \sum_{j=1}^{n-M}z_ju_j .   
\end{align} 
Thereby, $\{t_j\}_{j=1}^M$ is a basis of the polynomials $\mathcal{P}_{m-1}$ (as before) and $\{z_j\}_{j=1}^{n-M}$ is a basis of the orthogonal complement $\mathcal{P}_{m-1}^{\perp}=\{f\in \mathcal{S}_n:\sum_{i=1}^n f(\xvec_i)p(\xvec_i)=0,\ p \in\mathcal{P}_{m-1}\}$. Moreover, the coefficient vectors follow an improper uniform prior $\alphavec\sim 1$ and a multivariate normal prior $\uvec\sim N(\nullvec,\tau^2 \Kmat_z^{-1})$, where $\Kmat_z=\{\mathcal{B}_{m}(z_j,z_l)\}$ is the roughness penalty matrix in terms of the basis $\{z_j\}_{j=1}^{n-M}$.  

The prior~\eqref{BayesianSmoothnessPrior} is a sum of an 
improper uniform prior on $\mathcal{P}_{m-1}$ and a Gaussian process (GP) prior on $\mathcal{P}^\perp_{m-1}$. The GP has mean function zero and covariance function $c(\xvec,\widetilde{\xvec})=\tau^2 k(\xvec,\widetilde{\xvec})$ with $k(\xvec,\widetilde{\xvec})=\zvec(\xvec)^t \Kmat_z^{-1}\zvec(\widetilde{\xvec})$, where $\zvec(\cdot)=(z_1(\cdot),\dots,z_{n-M}(\cdot))^t$ is the column vector of basis function evaluations.
For smoothing variance $\tau^2=\sigma^2/(\lambda n)$ the posterior mode coincides with $\widehat{f}_\lambda$, which justifies to regard~\eqref{BayesianSmoothnessPrior} as equivalent Bayesian smoothness prior or thin plate spline prior \citep[cf.][Section 6]{silverman_aspects_1985}. Proposition~\ref{PropositionConstructionDRB} is important for our construction of the TDRB smoother.

\begin{proposition}\label{PropositionConstructionDRB}
    \begin{enumerate}[i)]
    \item $k(\xvec,\widetilde{\xvec})$ is the reproducing kernel of the finite-dimensional reproducing kernel Hilbert space $\mathcal{P}_{m-1}^\perp$ with inner product $\mathcal{B}_m$, i.e., $\mathcal{B}_m(k(\cdot,\widetilde{\xvec}),f)=f(\widetilde{\xvec}),\ f\in \mathcal{P}_{m-1}^\perp,\ \widetilde{\xvec}\in \mathds{R}^d$.
     \item The $n\times n$ kernel matrix $\Kmat=\{k(\xvec_i,\xvec_j)\}$ can be expressed in the form $\Kmat=\Pmat\Emat\Pmat$, where $\Pmat=\Imat_n-\Tmat(\Tmat^t\Tmat)^{-1}\Tmat^t$ is the orthogonal projector onto $span(\Tmat)^\perp$ with $n\times n$ identity matrix $\Imat_n$. 
     \item If $\{\varphi_j\}_{j=1}^n$ is a Demmler-Reinsch basis, then the vectors of function evaluations of the nonpolynomial basis functions $\{\varphi_j\}_{j=M+1}^n$ are eigenvectors of $\Kmat$. Conversely, the nonpolynomial basis functions of a Demmler-Reinsch basis can be constructed through an eigendecomposition of $\Kmat$.
    \end{enumerate}
\end{proposition}

 See Appendix~\ref{sec:ProofReproducing} for a proof. 
 The key step of our construction of the TDRB smoother is thus an eigendecomposition of the kernel matrix $\Kmat=\Pmat\Emat\Pmat$. As we only require the first $k$ basis functions, a truncated eigendecomposition is sufficient. In conclusion, we obtain Algorithm~\ref{algo1} summarized below.
 
\begin{algorithm}\label{algo1}
\caption{Construction of the TDRB smoother with computational cost $O(n^2k)$.
}\label{alg:cap}
~
\KwIn{Input matrix $\Xmat=(\xvec_1,\dots,\xvec_n)^t$, penalty order $m$, number of basis functions $k$}
\KwOut{Design matrix $\Phimat_k$ of size $n\times k$, diagonal penalty matrix $\Gammamat_k$ of size $k\times k$, transition matrix $\Amat_k$ of size $(n+M) \times k$} 
\textbf{Construction:}
\begin{enumerate}
\item Set up the $n\times M$ polynomial design matrix $\Tmat$ and the $n\times n$ RBF design matrix $\Emat$.
    \item Compute a QR decomposition $\Tmat=\Qmat\Rmat$ with $\Qmat$ of size $n\times M$ and $\Rmat$ of size $M\times M$.
    \item Compute the kernel matrix $\Kmat=\Pmat \Emat \Pmat$ efficiently, where $\Pmat=\Imat_n-\Qmat \Qmat^t$ (see Remark~\ref{RemarkAlgo1}).
    \item Compute a truncated eigendecomposition $\Kmat\approx \Umat_{k-M}\Lambdamat_{k-M}\Umat_{k-M}^t$. Thereby $\Lambdamat_{k-M}$ is diagonal and has the largest $(k-M)$ eigenvalues of $\Kmat$ as diagonal entries (sorted decreasingly) and $\Umat_{k-M}$ has the corresponding orthonormal eigenvectors as columns.
    \item Compute the transition matrix 
    $\Amat_k=\sqrt{n}\begin{pmatrix}
   \Rmat^{-1} & -\Rmat^{-1}\Qmat^t \Emat\Pmat\Umat_{k-M}\Lambdamat_{k-M}^{-1}\\ \nullmat & \Pmat\Umat_{k-M}\Lambdamat_{k-M}^{-1}
\end{pmatrix}.$
\item Return the design matrix $\Phimat_k=\sqrt{n}(\Qmat,\Umat_{k-M})$, the penalty matrix $\Gammamat_k=\text{bdiag}(\nullvec_M,n\Lambdamat_{k-M}^{-1})$,\\ the transition matrix $\Amat_k$ and the proportion of prior variance explained $\text{tr}(\Lambdamat_{k-M})/\text{tr}(\Kmat)$.
\end{enumerate}
\end{algorithm}

\begin{remark}\label{RemarkAlgo1}Some remarks are as follows.
\begin{enumerate}[a)]
    \item Computational cost: The matrices $\Pmat$ and $\Emat$ are both of size $n\times n$ so that direct computation of $\Pmat \Emat \Pmat$ has cubic cost $O(n^3)$. However, using that $\Pmat\Emat\Pmat=\Emat-\Hmat\Emat-\Emat\Hmat+\Hmat\Emat\Hmat$ with $\Hmat=\Qmat\Qmat^t$ one can reduce the cost to $O(n^2M)$. To this end, first compute $\Emat\Hmat=\Emat\Qmat\Qmat^t$. Then set $\Hmat\Emat=(\Emat\Hmat)^t$ and use this to compute $\Hmat\Emat\Hmat=\Hmat\Emat\Qmat\Qmat^t$. Using Lanczos algorithm, the truncated eigendecomposition of $\Pmat\Emat\Pmat$ has cost $O(n^2(k-M))$ \citep{wood_thin_2003}. Thus, the overall cost  of Algorithm~\ref{algo1} is $O(n^2k)$, which is exactly the same as required to set up TPRS \citep{wood_thin_2003}.
    \item Replicated inputs: In practice often replicated inputs occur,  i.e., $\xvec_i=\xvec_j$ for some $i\neq j$, especially for large sample sizes $n$. While Algorithm~\ref{algo1} can be used, a more efficient approach is to introduce a diagonal weight matrix $\Wmat$ that accounts for the multiplicity of the inputs. With this, one can reduce the cost to $O(n^2_{unique}k)$, where $n_{unique}$ is the number of unique inputs. Often $n_{unique}$ is much smaller than $n$, so that this can make a substantial difference; see Appendix~\ref{sec:Algo2} for details.
        \item Predictions at test points: The matrix $\Amat_k$ is called transition matrix since $\Phimat_k=(\Tmat,\Emat)\Amat_k$. The matrix $\Amat_k$ facilitates predictions at test points $\xvec_1^\ast,\dots,\xvec_{n^\ast}^\ast\in \mathds{R}^d$. To this end, simply set up the corresponding polynomial and RBF design matrices $\Tmat^\ast$ and $\Emat^\ast$ and then compute $\Phimat_k^\ast=(\Tmat^\ast,\Emat^\ast)\Amat_k$.

   \item Connection to KL expansion: The eigenvectors of $\Kmat$ correspond to eigenfunctions of the operator $T_n:\mathcal{P}^\perp_{m-1}\to \mathcal{P}^\perp_{m-1}$ with $T_nf(\cdot)=1/n\sum_{i=1}^n k(\cdot,\xvec_i)f(\xvec_i)$. 
    This reveals a close connection between the Demmler-Reinsch basis and a KL expansion \citep[see, e.g.,][Section 1.4]{ash1975topics}. More specifically, the nonpolynomial basis functions $\{\varphi_j\}_{j=M+1}^n$ of a Demmler-Reinsch basis coincide with the basis functions of a discrete KL expansion, i.e., a KL expansion based on the empirical distribution of the inputs $\xvec_i,\ i=1,\dots,n$, instead of Lebesgue measure. The relevant optimality properties are addressed in Proposition~\ref{Prop:BayesianOptimality} further below.

    \item Infinite-dimensional Bayesian formulation:  \cite{silverman_aspects_1985} distinguishes a finite-dimensional and an infinite-dimensional Bayesian formulation for smoothing splines. This distinction is also valid for thin plate splines and the prior~\eqref{BayesianSmoothnessPrior} corresponds to the finite-dimensional formulation. The infinite-dimensional formulation uses a GP with covariance function $\widetilde{c}(\xvec,\widetilde{\xvec})=\tau^2\widetilde{k}(\xvec,\widetilde{\xvec})$, where $\widetilde{k}(\xvec,\widetilde{\xvec})$ is the RK of the RKHS $\{f\in \text{BL}^m(\mathds{R}^d):\sum^n_{i=1} f(\xvec_i)p(\xvec_i)=0,\ p\in \mathcal{P}_{m-1}\}$; see \citet[][Section 4.3.2]{gu_smoothing_2013} for details on $\widetilde{k}(\xvec,\widetilde{\xvec})$. While $\widetilde{k}\neq k$, one can show that $\widetilde{k}(\cdot,\xvec_i)=k(\cdot,\xvec_i),\ i=1,\dots,n$. Thus, a discrete KL expansion of the GP with covariance $\widetilde{c}(\xvec,\widetilde{\xvec})$ leads to the same basis functions.
\end{enumerate}
\end{remark}

From a Bayesian perspective, using the TDRB smoother $\widehat{f}_\lambda^k$ corresponds to replacing the thin plate spline prior~\eqref{BayesianSmoothnessPrior} by the prior
\begin{align}\label{TDRBPrior}
\widetilde{f}=\widetilde{f}_{poly}+\widetilde{f}_{gp}=\sum_{j=1}^M\varphi_j\beta_j+\sum_{j=M+1}^{k}\varphi_{j}\beta_j.
\end{align}
Thereby, $\{\varphi_j\}_{j=1}^k$ are the first $k$ basis functions of a Demmler-Reinsch basis and the coefficients have priors $\beta_j\sim 1,\ j=1,\dots,M,$ and $\beta_j\sim N(0,\tau^2/\gamma_{j}),\ j=M+1,\dots,k$. Proposition~\ref{Prop:BayesianOptimality} investigates the relation of~\eqref{TDRBPrior} and the thin plate spline prior~\eqref{BayesianSmoothnessPrior}.

\begin{proposition}~\label{Prop:BayesianOptimality} Let $\{\varphi_j\}_{j=1}^n$ be a Demmler-Reinsch basis. For $k\in\{M+1,\dots,n\}$ let $H_{\mathcal{U}^\ast}$ be the least squares projection onto the $(k-M)$-dimensional subspace $\mathcal{U}^\ast=span(\{\varphi_j\}_{j=M+1}^k)\subseteq\mathcal{P}^\perp_{m-1}$. Let further $H_\mathcal{U}$ be the least squares projection onto an arbitrary $(k-M)$-dimensional subspace $\mathcal{U}\subseteq\mathcal{P}^\perp_{m-1}$ and let $\Hmat_{\mathcal{U}}$ be the corresponding $n\times n$ projection matrix. Let finally $f_{gp}$ as in~\eqref{BayesianSmoothnessPrior}. Then it holds:
\begin{enumerate}[i)]
    \item $E\|f_{gp}-H_{\mathcal{U}^\ast}f_{gp}\|^2_n\leq E\|f_{gp}-H_{\mathcal{U}}f_{gp}\|^2_n$.
    \item $E\|f_{gp}\|^2_n\geq E\|H_{\mathcal{U}^\ast}f_{gp}\|^2_n\geq E\|H_{\mathcal{U}}f_{gp}\|^2_n$ as well as $E\|f_{gp}\|^2_n=\tau^2/n\ \text{tr}(\Kmat),\ E\|H_{\mathcal{U}^\ast}f_{gp}\|_n^2=\tau^2/n\ \text{tr}(\Lambdamat_{k-M})$ and $E\|H_{\mathcal{U}}f_{gp}\|_n^2=\tau^2/n\ \text{tr}(\Hmat_{\mathcal{U}}\Kmat)$.
    \item The distribution of the projected process $H_{\mathcal{U}^\ast} f_{gp}$ is the same as that of $\widetilde{f}_{gp}$ in~\eqref{TDRBPrior}. 
\end{enumerate}
\end{proposition}
 \begin{proof}
    See Appendix~\ref{sec:ProofBayesianOptimality}.
\end{proof}

Proposition~\ref{Prop:BayesianOptimality} shows that~\eqref{TDRBPrior} is an optimal low rank approximation of the thin plate spline prior~\eqref{BayesianSmoothnessPrior}. In addition to that, Proposition~\ref{Prop:BayesianOptimality} justifies to use the ratio $E\|H_{\mathcal{U}^\ast}f_{gp}\|_n^2/E\|f_{gp}\|^2_n=\text{tr}(\Lambdamat_{k-M})/\text{tr}(\Kmat)$  as a measure of the explained prior variance for the TDRB smoother.

The following Section~\ref{sec:EmpiricalEvidence} provides empirical evidence for the TDRB smoother.


\section{Empirical evidence}\label{sec:EmpiricalEvidence}
This section provides empirical evidence. We first consider simulations and then a real data example.
\subsection{Simulations}
We consider two-dimensional spatial smoothing using the data generating process
    \begin{align*}
        y_i=f_0(x_{i1},x_{i2})+\epsilon_i=\sin(2\pi x_{i1})\cos(2\pi x_{i2})+\epsilon_i,\ \epsilon_i\overset{iid}{\sim} N(0,\sigma^2),\ i=1,\dots,n,
    \end{align*}
    for sample size $n=1000$, error variance $\sigma^2=1$ and uniformly distributed inputs $\xvec_i\overset{iid}{\sim} U([0,1]^2)$. We generate $R=1000$ replicate data sets and use the following low rank thin plate spline approximations with second order penalization $(m=2)$ for estimation of $f_0$: 
    \begin{enumerate}
        \item \textbf{RSUB:} We use a random subset of the inputs $\xvec_1,,\dots,\xvec_n$ as knots $\kappavec_1,\dots,\kappavec_k$
        and then minimize the penalized sum of squares~\eqref{PenalizedSumSquares} over the corresponding space of natural thin plate splines;
        see \citet[][p. 219-221]{wood_generalized_2017} for further details.
        \item \textbf{TPRS:} We follow the description in  \citet[][Appendix A]{wood_thin_2003} to set up the TPRS design matrix and penalty matrix. To compute the truncated eigenvalue decomposition of the RBF design matrix $\Emat$, we use the function \texttt{slanczos} in the \texttt{R} package \texttt{mgcv}. 
        \item \textbf{TDRB:} We use Algorithm~\ref{algo1} for the construction of the TDRB smoother. To compute the truncated eigendecomposition of $\Kmat=\Pmat\Emat\Pmat$, we again use the function \texttt{slanczos}.
    \end{enumerate}
For each of the methods we vary the number of basis functions $k\in \{5,6,\dots,19\}\cup\{20,30,\dots,80\}$. The idea is to cover both, a relatively small number of basis functions $(k\leq 19)$ and a larger number of basis functions, which seems more suitable in practice ($k\geq 20$). 

To infer a suitable value $\widehat{\lambda}$ we use generalized crossvalidation (GCV). Following \cite{craven1979smoothing} the GCV score is defined as $\text{GCV}(\lambda)={1/n\ \|\yvec-\Smat_\lambda^k \yvec\|^2}/{[1-\text{tr}(\Smat_\lambda^k)/n]^2}$ and minimized over a fine grid of $\lambda$ values. To speed up the computation, we exploit that the GCV score has a computationally convenient expression for biorthogonal basis functions. More specifically, if the Gramian matrix $\Xmat^t\Xmat/n$ is diagonal with diagonal elements $g_j,\ j=1,\dots,k,$ and the penalty matrix $\Kmat_x$ is diagonal with diagonal elements $\gamma_j, j=1,\dots,k$, then the GCV score can be expressed in the form
     \begin{align}\label{GCVBiorthogonal}
         \text{GCV}(\lambda)=\dfrac{\yvec^t\yvec/n+\sum_{j=1}^k \widehat{\beta}_j^2 \{g_jd_j(\lambda)^2-2d_j(\lambda)\}}{[1-1/n \sum^k_{j=1} g_jd_j(\lambda)]^2},
     \end{align}
where $d_j(\lambda)=1/(g_j+\lambda \gamma_j)$ and $\widehat{\betavec}=\Xmat^t\yvec/n$. Expression~\eqref{GCVBiorthogonal} follows from straightforward computations and implies that the GCV score can be computed very efficiently in $O(k)$ operations for each $\lambda$ value.
For the TDRB smoother biorthogonality is directly given from the definition. To use expression~\eqref{GCVBiorthogonal} also for RSUB and TPRS, we use a reparametrization. This reparametrization is based on solving the generalized eigenvalue problem~\citep[][Chapter 15]{parlett_symmetric_1998} for the pair of matrices $(\Kmat_x,\Xmat^t\Xmat/n+\Kmat_x)$ and yields biorthogonal basis functions. The reparametrization not only facilitates computation of the GCV score but also decomposes the design matrix of RSUB and TPRS into a polynomial part and the orthogonal complement. Thus, we can use the ratio $\text{tr}(\Hmat_\mathcal{U}\Kmat)/\text{tr}(\Kmat)$ as proportion of explained variance, where $\Hmat_{\mathcal{U}}$ is the projection matrix onto the nonpolynomial part (cf. Proposition~\ref{Prop:BayesianOptimality}). 

\paragraph{Results.}
Figure~\ref{fig:RMSE} shows the root mean squared error $\text{RMSE}=\|\widehat{f}_{\widehat{\lambda}}^k-f_0\|_n$. Table~\ref{Tab:PPVE} shows the proportion of prior variance explained. The main observations are as follows: 
\begin{itemize}
    \item For a small number of basis functions ($k\leq 19$) the performance of the three smoothers is quite different, but there is no clear pattern. In particular, there are also nonnegligible differences between TPRS and TDRB.
    \item For a larger number of basis functions ($k\geq 20$) the performance of the TDRB smoother and TPRS is very similar and consistently better than that of RSUB.
    \item In line with Proposition~\ref{Prop:BayesianOptimality} the proportion of prior variance explained is always largest for the TDRB smoother. However, TPRS capture almost the same amount of prior variance, especially when a larger number of basis functions is used. 
\end{itemize}

\begin{figure}[H]
    \centering
    \includegraphics[width=\linewidth]{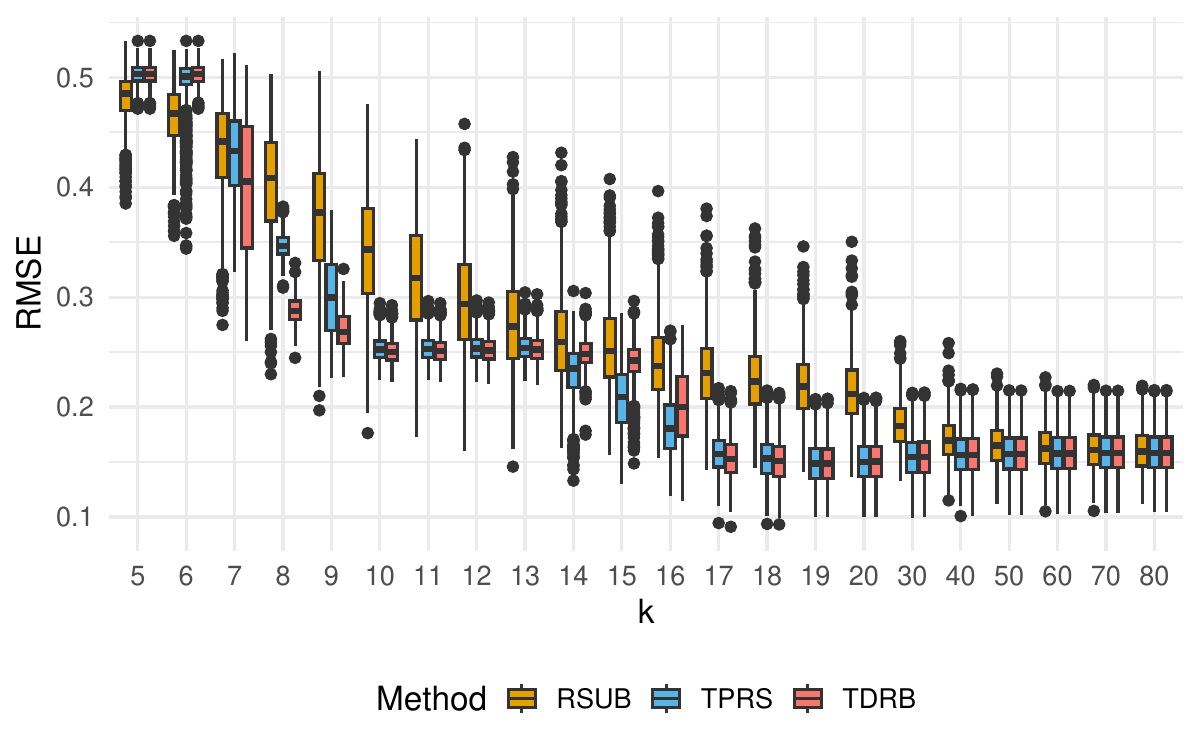}
    \caption{Root mean squared error (RMSE) in simulations. Lower is better and $k$ is the number of basis functions that were used for the low rank thin plate spline approximations.} 
    \label{fig:RMSE}
\end{figure}

\begin{table}[H]
\centering
\begin{tabular}{lrrr|lrrr}
  \hline
k & RSUB & TPRS & TDRB & k & RSUB & TPRS & TDRB \\ 
  \hline
5 & 33.95 & 47.22 & 47.22 & 16 & 84.20 & 89.14 & 89.17 \\ 
  6 & 46.91 & 60.53 & 61.34 & 17 & 85.27 & 90.03 & 90.04 \\ 
  7 & 56.54 & 67.62 & 67.86 & 18 & 86.17 & 90.68 & 90.69 \\ 
  8 & 63.09 & 73.61 & 73.91 & 19 & 87.01 & 91.30 & 91.30 \\ 
  9 & 68.08 & 77.25 & 77.39 & 20 & 87.76 & 91.77 & 91.78 \\ 
  10 & 72.15 & 80.64 & 80.65 & 30 & 92.36 & 94.92 & 94.92 \\ 
  11 & 75.15 & 82.77 & 82.79 & 40 & 94.53 & 96.38 & 96.38 \\ 
  12 & 77.60 & 84.67 & 84.69 & 50 & 95.78 & 97.21 & 97.21 \\ 
  13 & 79.90 & 86.15 & 86.16 & 60 & 96.57 & 97.75 & 97.75 \\ 
  14 & 81.54 & 87.20 & 87.24 & 70 & 97.14 & 98.12 & 98.12 \\ 
  15 & 82.72 & 88.19 & 88.24 & 80 & 97.56 & 98.39 & 98.39 \\ 
   \hline
\end{tabular}
\caption{Proportion of prior variance explained in percent. Shown is the mean across the $1000$ replicates (rounded to two decimal places). A higher value indicates a better approximation of the thin plate spline and $k$ is the number of basis functions that were used for the low rank approximations. } 
 \label{Tab:PPVE}
\end{table}

\newpage
\subsection{Real data example: Bayesian additive modelling of mackerel egg data}\label{sec:RealDataExample}
To demonstrate that the TDRB smoother is also applicable for additive models, we reanalyze the mackerel egg data similar to \citet[][Section 4]{wood_thin_2003}. The data set \texttt{mack} is part of the \texttt{R} package \texttt{gamair} and consists of $n=634$ observations. The response is egg density and the covariates of interest are distance, depth, longitude and latitude; see \cite{wood_thin_2003} and references therein for more information about the data. Just as in \citet[][]{wood_thin_2003} we consider the model
\begin{align*}
    \sqrt{egg.dens}=\beta_0+f_{dist}(dist)+f_{depth}(depth)+f_{spat}(lon,lat)+\epsilon,\epsilon\sim N(0,\sigma^2).
\end{align*}
In contrast to \cite{wood_thin_2003} we use the TDRB smoother instead of TPRS for the three additive components $f_{dist}, f_{depth}, f_{spat}$. In addition, we use a fully Bayesian approach with weakly informative inverse gamma priors $\text{IG}(1/1000,1/1000)$ for all variance parameters. The penalty order is $m=2$ and the numbers of basis functions are $k=10$ for $f_{dist}$ and $f_{depth}$ as well as $k=50$ for $f_{spat}$. This yields a proportion of prior variance explained of 99.76\%, 99.85\% and 98.43\%, respectively. For posterior sampling we use a Gibbs sampler with $10^5$ iterations \citep[see, e.g., ][for details]{fahrmeir_regression_2021}. 

To enhance numerical stability, we use a preprocessing step and first transform the raw inputs of $dist$ and $depth$ to the unit interval $[0,1]$ and the spatial coordinates $lon$ and $lat$ to the unit cube $[0,1]^2$ (using the same scaling for $lon$ and $lat$ to maintain isotropy). Such a preprocessing step is standard and the results are transformed back to the original scale after estimation. To solve the identifiablity issue inherent in additive models, we use the monomial basis with $t_1\equiv 1$ for the polynomial nullspace of each additive component. With this, the first Demmler-Reinsch basis function of each smooth is constant (due to the transformation to the unit interval/cube, this is still true if pivoting is used for the QR decomposition of $\Tmat$). This means that we can simply omit the first column of the design matrix of each smooth to ensure identifiability. 

Figure~\ref{fig:mackerel} shows the estimated posterior mean effects $\widehat{f}_{dist}$, $\widehat{f}_{depth}$ and $\widehat{f}_{spat}$, which closely resemble the estimated effects of \citet[][Section 4]{wood_thin_2003} based on TPRS.
\begin{figure}[H]
    \centering
    \includegraphics[width=\linewidth]{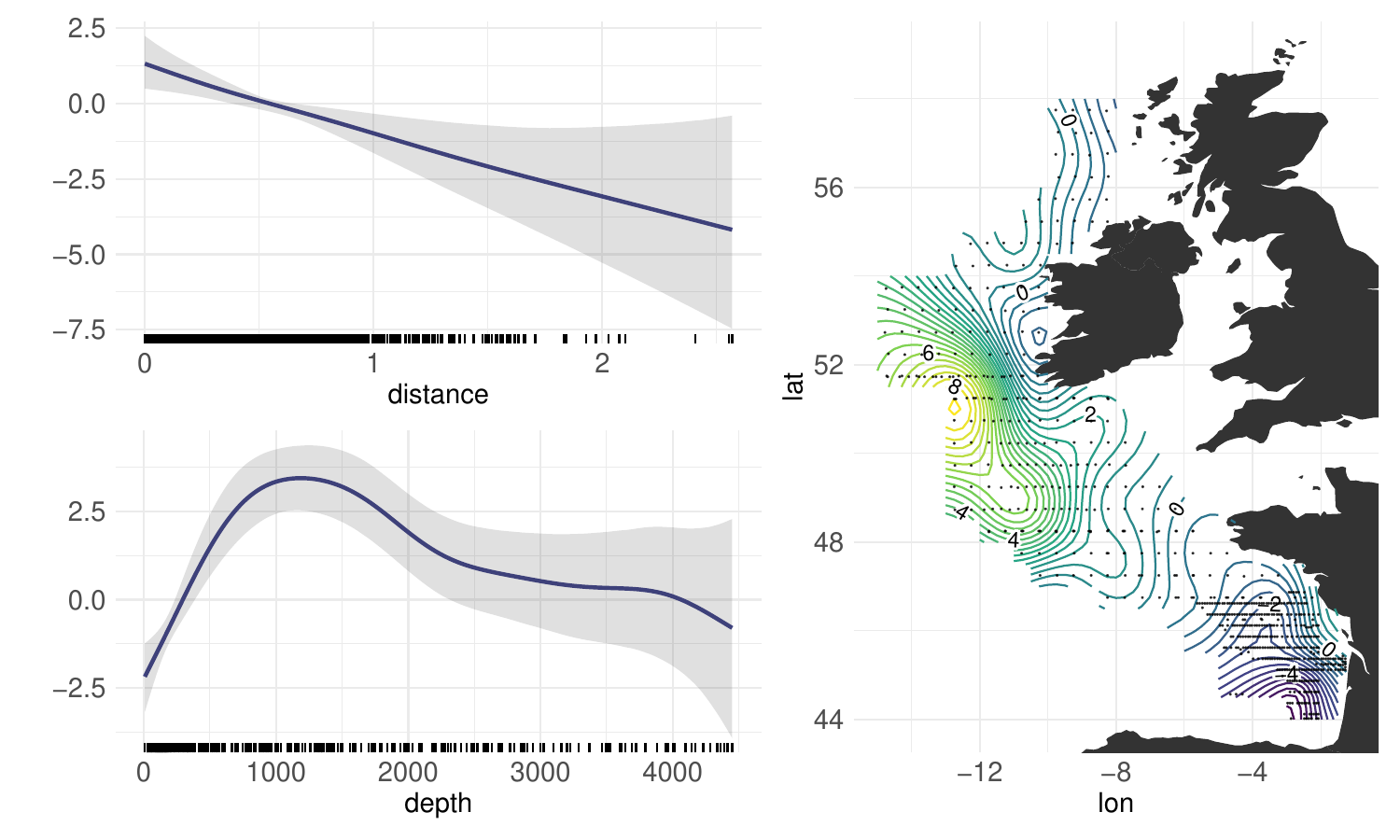}
    \caption{TDRB smoother for mackerel egg data. Shown are the estimated posterior mean effects of distance and depth with pointwise 95\% credible bands (left) as well as a contour plot of the estimated spatial effect (right).}
    \label{fig:mackerel}
\end{figure}


\section{Discussion}~\label{sec:Discussion}
This paper has introduced a computationally efficient low rank approximation of thin plate splines. The key idea is to restrict the minimization of the penalized sum of squares~\eqref{PenalizedSumSquares} to the subspace that is spanned by the first $k$ basis functions of a Demmler-Reinsch basis and we refer to the resulting smoother as TDRB smoother.
In our simulations the performance of the TDRB smoother was very similar to that of TPRS unless a very small number of basis functions $k$ was used. Given their close relation, this is not surprising. As demonstrated in the real data example of Section~\ref{sec:RealDataExample}, the TDRB smoother can be used as main building block for estimation of GAMs similar to TPRS. Despite their similarities, we argue that the TDRB smoother has some advantages:

\begin{itemize}
    \item Optimality: First, the optimality of the TDRB smoother is more transparent as that of TPRS. For uni- or multivariate scatterplot smoothing, we have a best low rank approximation of the smoother matrix in terms of Frobenius norm and spectral norm. From a Bayesian perspective, we have a best low rank approximation of the (proper part of) the equivalent Bayesian smoothness prior in the sense of a discrete KL expansion. In contrast to that, TPRS only provide an optimal approximation if one ignores the constraints $\Tmat^t\deltavec=\nullvec$.
    \item Choice of $k$: \citet[][]{wood_thin_2003} suggests to use the effective degrees of freedom to decide if the number of basis functions $k$ is large enough. For the TDRB smoother this quantity can be complemented by the proportion of prior variance explained  $\text{tr}(\Lambdamat_{k-M})/\text{tr}(\Kmat)$. A too low value (say $< 0.9$) may indicate that a higher number of basis functions $k$ is required.
    While it is also possible to compute the proportion of prior variance explained for TPRS, this requires to first compute the kernel matrix $\Kmat=\Pmat\Emat\Pmat$. Using the RBF design matrix $\Emat$ to compute the proportion of prior variance explained is not possible: The diagonal elements of $\Emat$ are equal to $0$ so that $\text{tr}(\Emat)=0$. 
    \item Biorthogonality: The biorthogonality of the Demmler-Reinsch basis functions can be be used to speed up computations as demonstrated with the GCV score. In addition, it is very straightforward to use the TDRB smoother in additive models. This is because one can easily ensure that the first Demmler-Reinsch basis function is constant. The remaining basis functions are orthogonal so that one can simply omit the first column of the design matrix to ensure identifiability. Thus, in contrast to TPRS one does not need another reparametrization.
    \item Asymptotics: As we have shown, it is fairly straightforward to prove rigorous asymptotic results for the TDRB smoother by combining existing results for thin plate splines with proof strategies developed for the asymptotics of penalized B-splines. In contrast to that, it seems much more difficult to establish rigorous asymptotic results for TPRS.
\end{itemize}

\citet[][Section 5]{wood_thin_2003} posed several open questions for TPRS. Among these are the question whether a more strongly optimal approximation of thin plate splines is possible and whether additional guidance regarding the choice of the number of basis functions $k$ can be given. This paper has demonstrated that a truly optimal approximation is possible and provides additional guidance on the choice of $k$ through the proportion of prior variance explained. In conclusion, we think that the TDRB smoother can be regarded as a slightly refined version of TPRS and provides a valuable option for estimation of multivariate functions and GAMs. 

\bibliographystyle{apalike}
\bibliography{references}


\appendix
\section{Appendix}\label{sec:Appendix}

\subsection{Proof of Proposition 1 (Frequentist optimality)}\label{sec:ProofTheo1}

\begin{proof} 
 Let $\{\varphi_j\}_{j=1}^n$ be a Demmler-Reinsch basis of $\mathcal{S}_n$ and let $\varphivec_j,\ j=1,\dots,n,$ denote the corresponding vectors of function evaluations, i.e., $\varphivec_j=(\varphi_j(\xvec_1),\dots,\varphi_j(\xvec_n))^t,\ j=1\dots,n$.
 In terms of $\{\varphi_j\}_{j=1}^n$, the smoother matrix of the thin plate spline $\widehat{f}_\lambda$ is $\Smat_\lambda=\sum_{j=1}^n \varphivec_j\varphivec_j^t \frac{1}{n+n\lambda \gamma_j}$.
 Furthermore, the smoother matrix of the TDRB smoother $\widehat{f}_\lambda^k$ is $\Smat^k_\lambda=\sum_{j=1}^k \varphivec_j\varphivec_j^t \frac{1}{n+n\lambda \gamma_j}$. Close inspection of the two expressions reveals that $\Smat_\lambda^k$ is a truncated singular value decomposition of $\Smat_\lambda$. The statement directly follows by the Eckart-Young theorem \citep[see, e.g.,][Section 2.4]{golub_matrix_2013}.
\end{proof}

\subsection{Proof of Theorem 1 (Asymptotics)}\label{sec:ProofAsymptotics}

First, we introduce some notation. For each $n\in \mathds{N}$ we have:

\begin{itemize}
    \item A vector of observations $\yvec^{(n)}=(y_1,\dots,y_n)^t$, a vector of residuals $\epsilonvec^{(n)}=(\epsilon_1,\dots,\epsilon_n)^t$ and a vector of function evaluations $\fvec_0^{(n)}=(f_0(\xvec_1),...,f_0(\xvec_n))^t$. These are related via: $\yvec^{(n)}=\fvec_0^{(n)}+\epsilonvec^{(n)}$.
    \item A Demmler-Reinsch basis $\{\varphi_{j,n}\}_{j=1}^n$ with design matrix $\Phimat_n$ and penalty matrix $\Gammamat_n=\text{diag}(\gamma_{1,n},\dots,\gamma_{n,n})$.
    \item A pseudo-true parameter vector $\betavec_0^{(n)}=\Phimat_n^t\fvec_0^{(n)}/n$.
    \item A vector of estimated function values $\widehat{\fvec}^{k_n}_{\lambda_n}=\Smat_{\lambda_n}^{k_n} \yvec^{(n)}$ with smoother matrix $\Smat^{k_n}_{\lambda_n}=\Phimat_n \Dmat_{\lambda_n}^{k_n}\Phimat_n^t/n$. Thereby, $\Dmat_{\lambda_n}^{k_n}$ is a diagonal matrix with diagonal entries $1/(1+\lambda_n\gamma_{j,n})$ for $j=1,\dots,k_n$ and $0$ for $j=k_n+1,\dots,n$.
\end{itemize}

We use two lemmas which follow directly from classical results for thin plate splines. 
\begin{lemma}[Eigenvalue bounds]\label{LemmaEigenvalueBounds}
    There exist constants $\alpha,\beta>0$ such that 
\begin{align}\label{eigenvalueBounds}
   \alpha j^{2m/d} \leq \gamma_{j,n}\leq \beta j^{2m/d}
\end{align}
for $j=M+1,\dots,n$ and all sufficiently large $n$.
\end{lemma}
\begin{proof}
For each $n$ there exists a cardinal spline basis $\{c_{j,n}\}_{j=1}^n$ with design matrix $\Imat_n$ and penalty matrix $\Kmat_{card}$.
From \citet[Theorem 5.3]{utreras_convergence_1988} it follows that  conditions A2) and A3) are sufficient to have the bounds~\eqref{eigenvalueBounds} for the eigenvalues of the matrix $n\Kmat_{card}$. Thus, we only need to show that the eigenvalues of $n\Kmat_{card}$ coincide with the $\gamma_{j,n}$. To this end, apply a change of basis with transition matrix $\Phimat_n^t/n$ to the Demmler-Reinsch basis $\{\varphi_{j,n}\}_{j=1}^n$. The resulting basis has design matrix $\Phimat_n\Phimat_n^t/n=\Imat_n$ and is thus the cardinal spline basis. The corresponding penalty matrix is $\Kmat_{card}=(\Phimat_n/n)\Gammamat_n(\Phimat_n^t/n)$. Hence $n\Kmat_{card}=(\Phimat_n/\sqrt{n})\Gammamat_n(\Phimat_n^t/\sqrt{n})$. On the right hand side we recognize an eigendecomposition so that $n\Kmat_{card}$ has indeed eigenvalues $\gamma_{j,n}$.
\end{proof}
    
\begin{lemma}[Roughness bound]\label{LemmaBoundedRoughness}
   There exists a constant $K$ such that 
\begin{align*}
    (\betavec_0^{(n)})^t\Gammamat_n\betavec_0^{(n)}\leq K \|f_0\|_{\Omega}^2,
\end{align*} where $\|f_0\|^2_\Omega=\sum_{|\nuvec|\leq m}\int_{\Omega} \{D^{\nuvec} f(\xvec)\}^2 d\xvec$ is the squared Sobolev-norm of $f_0$.
\end{lemma}
\begin{proof}
From \citet[][Lemma 3.1]{duchon1978erreur} it follows that $f_0$ has a unique extension $f_0^\Omega$ to the Beppo Levi space $\text{BL}^m(\mathds{R}^d)$. This extension is the solution to the problem
\begin{align*}
    \underset{u\in \text{BL}^m(\mathds{R}^d), u|_\Omega=f}{\min} J_m(u),
\end{align*}
and there exists a constant $K>0$ such that 
\begin{align}\label{ExtensionBound}
    J_m(f_0^\Omega)\leq K \|f_0\|_\Omega^2.
\end{align}
It holds $\Phimat_n\betavec^{(n)}_0=\Phimat_n\Phimat_n^t\fvec_0^{(n)}/n=\fvec_0^{(n)}$, which implies that $\betavec_0^{(n)}$ is the coefficient vector of an interpolating natural thin plate spline. From~\cite{meinguet1979multivariate} it follows that the interpolating natural thin plate spline is unique and the smoothest function in $\text{BL}^m(\mathds{R}^d)$ passing through $\fvec_0^{(n)}$. Thus, $(\betavec_0^{(n)})^t\Gammamat_n\betavec_0^{(n)}\leq J_m(f_0^\Omega)$ and with~\eqref{ExtensionBound} the statement follows. 
\end{proof}

Next, we prove Theorem~\ref{TheoAsymptotics}.
\begin{proof} 
We decompose the error into the sum of a variance and a bias term
\begin{align*}
    E_0\|\widehat{f}^{k_n}_{\lambda_n}-f_0\|^2_n &=
    E_0\|\widehat{f}^{k_n}_{\lambda_n}-E_0\widehat{f}^{k_n}_{\lambda_n}+E_0\widehat{f}^{k_n}_{\lambda_n}-f_0\|^2_n=
    E_0\|\widehat{f}^{k_n}_{\lambda_n}-E_0\widehat{f}^{k_n}_{\lambda_n}\|^2_n+ \| E_0 \widehat{f}^{k_n}_{\lambda_n}-f_0\|^2_n.
\end{align*}
The variance term is $$E_0\|\widehat{f}^{k_n}_{\lambda_n}-E_0\widehat{f}^{k_n}_{\lambda_n}\|^2_n=E\|\Smat^{k_n}_{\lambda_n}\epsilonvec^{(n)}\|^2/n=\sigma^2/n\ \text{tr}(\{\Smat^{k_n}_{\lambda_n}\}^2)=\sigma^2/n \sum_{j=1}^{k_n} 1/(1+\lambda_n \gamma_{j,n})^2=O(\lambda_{n}^{-d/(2m)} /n),$$ where we used Lemma~\ref{LemmaEigenvalueBounds} and Lemma 9.1 of \citet[][]{gu_smoothing_2013} for the last equation. Next, we decompose the bias term into two summands
\begin{align*}
    \|E_0\widehat{f}_{\lambda_n}^{k_n}-f_0\|^2_n=\|\Smat_{\lambda_n}^{k_n}\fvec_0^{(n)}-\fvec_0^{(n)}\|^2/n=\sum_{j=1}^{k_n} \{1-1/(1+\lambda_n \gamma_{j,n})\}^2(\beta_{0,j}^{(n)})^2+\sum_{j=k_n+1}^n (\beta_{0,j}^{(n)})^2.
\end{align*}
The first summand is $O(\lambda_n)$ using that $x/(1+x)^2\leq 1/4$ for $x>0$ and that $(\betavec_0^{(n)})^t\Gammamat_n\betavec_0^{(n)}=O(1)$, which follows directly from Lemma~\ref{LemmaBoundedRoughness} and the assumption that $f_0$ is Sobolev-regular of order $m$. For the second summand it holds 
    $\gamma_{k_n+1,n} \sum_{j=k_n+1}^n (\beta_{0,j}^{(n)})^2 \leq (\betavec_0^{(n)})^t\Gammamat_n \betavec_0^{(n)} = O(1).$
Thus, using Lemma~\ref{LemmaEigenvalueBounds} the term is $O(k_n^{-2m/d})$. Altogether we thus have $$E_0\|\widehat{f}^{k_n}_{\lambda_n}-f_0\|^2_n=O(\lambda_n^{-d/(2m)}/n)+O(\lambda_n)+O(k_n^{-2m/d}).$$ For $\lambda_n\asymp n^{-2m/(2m+d)}$ and $k_n\asymp n^\delta$ with $\delta\in[d/(2m+d),n]$ we have $E_0\|\widehat{f}^{k_n}_{\lambda_n}-f_0\|^2_n=O(n^{-(2m)/(2m+d)})$.
\end{proof}

\subsection{Proof of Proposition 2 (Reproducing kernel)}

\begin{proof}\label{sec:ProofReproducing}
\begin{enumerate}[i)]
    \item Reproducing property: Let $f\in \mathcal{P}_{m-1}^\perp$ and $\widetilde{\xvec}\in \mathds{R}^d$. Since $\{z_j\}_{j=1}^{n-M}$ is a basis of $\mathcal{P}^\perp_{m-1}$, there exists a vector $\betavec\in \mathds{R}^{M-n}$ such that $f(\cdot)=\zvec(\cdot)^t\betavec$. 
    With this
    \begin{align*}
\mathcal{B}_m(k(\cdot,\widetilde{\xvec}),f(\cdot))=&\mathcal{B}_m(\zvec(\cdot)^t\Kmat_z^{-1}\zvec(\widetilde{\xvec}),\zvec(\cdot)^t\betavec)
=\zvec(\widetilde{\xvec})^t\Kmat_z^{-1}\Kmat_z\betavec=\zvec(\widetilde{\xvec})^t\betavec=f(\widetilde{\xvec}),
    \end{align*}
    where we used bilinearity of $\mathcal{B}_m$ and that $\Kmat_z=\{\mathcal{B}_{m}(z_j,z_l)\}$. 
    \item  Let $\Qmat_0$ be a matrix of size $n\times (n-M)$ whose columns form an orthonormal basis of the null space $ker(\Tmat^t)=span(\Tmat)^\perp\subset \mathds{R}^n$. 
With this, the orthogonal projector $\Pmat=\Imat_n-\Tmat(\Tmat^t\Tmat)^{-1}\Tmat^t$ has the form $\Pmat=\Qmat_0\Qmat_0^t$. Now consider the basis of $\mathcal{P}_{m-1}^\perp$ with design matrix $\Zmat=\Pmat\Emat\Qmat_0$ and penalty matrix $\Kmat_z=\Qmat_0^t\Emat\Qmat_0$. For this basis we have 
$$\Kmat=\Zmat\Kmat_z^{-1}\Zmat^t=\Pmat\Emat\Qmat_0 (\Qmat_0^t\Emat\Qmat_0)^{-1}\Qmat_0^t\Emat\Pmat=\Qmat_0\Qmat_0^t\Emat\Qmat_0(\Qmat_0^t\Emat\Qmat_0)^{-1}\Qmat_0^t\Emat\Pmat=\Qmat_0\Qmat_0^t\Emat\Pmat=\Pmat\Emat\Pmat.$$
\item Let $\{\varphi_j\}_{j=1}^n$ be a Demmler-Reinsch basis of $\mathcal{S}_n$ and let $\varphivec_j,\ j=1,\dots,n,$ denote the corresponding vectors of function evaluations, i.e., $\varphivec_j=(\varphi_j(\xvec_1),\dots,\varphi_j(\xvec_n))^t,\ j=1\dots,n$.
 In terms of $\{\varphi_j\}_{j=1}^n$, the kernel matrix has the form $\Kmat=\sum_{j=M+1}^n \varphivec_j\varphivec_j^t \frac{1}{\gamma_j}$, which implies that $\Kmat\varphivec_{j_0}=n/\gamma_{j_0}\varphivec_{j_0},\ j_0=M+1,\dots,n$.
Hence, the vectors of function evaluations of the nonpolynomial basis functions $\{\varphi_j\}_{j=M+1}^n$ are eigenvectors of $\Kmat$ with eigenvalues $\lambda_j=n/\gamma_j,\ j=M+1,\dots,n$.

It remains to show the other direction, i.e., that we can also construct the nonpolynomial basis functions of a Demmler-Reinsch basis through an eigendecomposition of $\Kmat$. 
To this end, compute an eigendecomposition $$\Zmat\Kmat_z^{-1}\Zmat^t=\Umat_{n-M}\Lambdamat_{n-M} \Umat_{n-M}^t.$$ Multiplying both sides by $\Umat\Lambdamat_{n-M}^{-1}\sqrt{n}$, this implies $\Zmat\Kmat_z^{-1}\Zmat^t\Umat\Lambdamat_{n-M}^{-1}\sqrt{n}=\sqrt{n}\Umat$. This motivates to consider the transition matrix $\Amat=\Kmat_z^{-1}\Zmat^t\Umat\Lambdamat_{n-M}^{-1}\sqrt{n}$ and the new basis of $\mathcal{P}^\perp_{m-1}$ with design matrix $\widetilde{\Zmat}=\Zmat\Amat$. This is an orthonormal basis since $$\widetilde{\Zmat}^t\widetilde{\Zmat}=\Amat^t\Zmat^t\Zmat\Amat=n\Umat^t\Umat=n\Imat_n.$$ In addition, the corresponding roughness penalty matrix is diagonal since $$\Amat^t\Kmat_z \Amat=\sqrt{n}\Lambdamat_{n-M}^{-1}\Umat^t\Zmat\Kmat_z^{-1}\Kmat_z\Kmat_z^{-1}\Zmat^t\Umat\Lambdamat_{n-M}^{-1}\sqrt{n}=n\Lambdamat_{n-M}^{-1}.$$ Hence, combining $\widetilde{\Zmat}$ with an orthonormal basis of the polynomials $\mathcal{P}_{m-1}$ we have a Demmler-Reinsch basis.
\end{enumerate}
\end{proof}

\subsection{Proof of Proposition 3 (Bayesian optimality)}\label{sec:ProofBayesianOptimality}
\begin{proof} 
Let $\mathcal{U}$ be a $(k-M)$-dimensional subpace of $\mathcal{P}_{m-1}^\perp$ and let $H_\mathcal{U}f=\text{argmin}_{u\in \mathcal{U}} \|f-u\|^2_n$ be the least squares projection onto $\mathcal{U}$. Let further $\Hmat_{\mathcal{U}}$ be the corresponding $n\times n$ projection matrix mapping $(f(\xvec_1),\dots,f(\xvec_n))^t$ onto $((H_{\mathcal{U}}f)(\xvec_1),\dots,(H_{\mathcal{U}}f)(\xvec_n))^t$.
Then it holds:
\begin{align}\label{ProjectionOrthogonality}
   0\leq  E\|f_{gp}-H_{\mathcal{U}}f_{gp}\|_n^2= E\|f_{gp}\|_n^2-2E(f_{gp},H_{\mathcal{U}}f_{gp})_n+E\|H_{\mathcal{U}}f_{gp}\|_n^2=E\|f_{gp}\|_n^2-E\|H_{\mathcal{U}}f_{gp}\|_n^2,
\end{align}
where we used that $(f_{gp},H_{\mathcal{U}}f_{gp})_n=(H_{\mathcal{U}}f_{gp},H_{\mathcal{U}}f_{gp})_n=\|H_{\mathcal{U}}f_{gp}\|_n^2$ for the last equality. This follows from the fact that $(f_{gp}-H_{\mathcal{U}}f_{gp},u)_n=0$, $u\in \mathcal{U}$. Equation \eqref{ProjectionOrthogonality} shows that $E\|H_{\mathcal{U}} f_{gp}\|_n^2\leq E\|f_{gp}\|_n^2$. By the same arguments as above we obtain 
\begin{align}\label{ProjectionOrthogonality2}
   0\leq  E\|f_{gp}-H_{\mathcal{U}^\ast}f_{gp}\|_n^2=E\|f_{gp}\|_n^2-E\|H_{\mathcal{U}^\ast}f_{gp}\|_n^2.
\end{align}
From~\eqref{ProjectionOrthogonality} and~\eqref{ProjectionOrthogonality2} we obtain
     $$E\|f_{gp}-H_{\mathcal{U}^\ast}f_{gp}\|_n^2\leq  E\|f_{gp}-H_{\mathcal{U}}f_{gp}\|_n^2 \iff E\|H_{\mathcal{U}}f_{gp}\|_n^2 \leq  E\|H_{\mathcal{U}^\ast}f_{gp}\|_n^2.$$
Hence, to prove i), it suffices to prove ii). 

By orthonormality of the Demmler-Reinsch basis functions the projection onto $\mathcal{U}^\ast$ can be expressed in the form $H_{\mathcal{U}^\ast}f_{gp}=\sum_{j=M+1}^k \varphi_j (\varphi_j,f_{gp})_n$.
From Proposition~\ref{PropositionConstructionDRB} we know that $\Kmat\varphivec_j=n/\gamma_j\varphivec_j=\lambda_j\varphivec_j$. Thus,  $(\varphi_j,f_{gp})_n=\varphivec_j^t\fvec_{gp}/n\sim N(0,\tau^2 \varphivec_j^t\Kmat\varphivec_j/n^2)=N(0,\tau^2 1/\gamma_j)$. This proves part iii). It remains to prove part ii). To this end, note first that
\begin{align*}
E\|H_{\mathcal{U}^\ast}f_{gp}\|_n^2=\sum_{j=M+1}^k E(\varphi_j,f_{gp})_n^2=\tau^2 \sum_{j=M+1}^k 1/\gamma_j=\tau^2/n\  \text{tr}(\Lambdamat_{k-M}).
\end{align*}
Obviously, $E\|f_{gp}\|_n^2=E\fvec_{gp}^t\fvec_{gp}/n= \tau^2 \text{tr}(\Kmat)/n$ and $
    E\|H_{\mathcal{U}}f_{gp}\|_n^2 = E\fvec_{gp}^t\Hmat_\mathcal{U}\fvec_{gp}/n= \tau^2 \text{tr}(\Hmat_\mathcal{U}\Kmat)/n.$
By \citet[][Theorem 1]{fan1949theorem} it holds $\text{tr}(\Hmat\Kmat)\leq \text{tr}(\Lambdamat_{k-M})$ for all projections $\Hmat$ onto $(k-M)$-dimensional subspaces of $\mathds{R}^n$ so that $E\|H_{\mathcal{U}}f_{gp}\|_n^2 \leq  E\|H_{\mathcal{U}^\ast}f_{gp}\|_n^2$. Since $\text{tr}(\Lambdamat_{k-M})\leq \text{tr}(\Lambdamat_{n-M})=\text{tr}(\Kmat)$, we also have 
$E\|H_{\mathcal{U}^\ast}f_{gp}\|_n^2\leq E\|f_{gp}\|_n^2$, which finishes the proof.
\end{proof}
 
\subsection{Construction of the TDRB smoother for replicated inputs}\label{sec:Algo2}
In practice often replicated inputs occur. While Algorithm~\ref{algo1} can still be used to construct the TDRB smoother, a more efficient approach is described below. 
With this, one can reduce the computational cost to $O(n^2_{unique}k)$ instead of $O(n^2k)$, where $n_{unique}$ is the number of unique inputs.
\begin{itemize}
    \item Given inputs $\widetilde{\xvec}_i,\ i=1,\dots,n,$ with replicates find the unique inputs $\xvec_i,\ i=1,\dots,n_{unique}$.
     \item Set up the diagonal weight matrix $\Wmat=\text{diag}(w_i)$, where $w_i=n_i/n$ for $i=1,\dots,n_{unique}$. Thereby, $n_i$ is the multiplicity of $\xvec_i$ in $\widetilde{\xvec}_i$ such that $\sum_{i=1}^{n_{unique}} w_i=1$.
    \item Apply Algorithm~\ref{alg:two} below to set up the TDRB smoother for the unique inputs $\xvec_i$.
    \item Replicate the corresponding rows of $\Phimat_k$ to obtain $\widetilde{\Phimat}_k$, the design matrix for $\widetilde{\xvec}_i$. With this, we have $\Phimat_k^t\Wmat\Phimat_k=\Imat_k$ and $\widetilde{\Phimat}_k^t\widetilde{\Phimat}_k=n\Imat_k$.
\end{itemize}

\begin{algorithm}[H]
\caption{Construction of the TDRB smoother with computational cost $O(n^2_{unique}k)$.}\label{alg:two}
~
\begin{enumerate}
\item Set up the polynomial design matrix $\Tmat$ and the RBF design matrix $\Emat$.
    \item Compute a QR decomposition of $\sqrt{\Wmat}\Tmat$, that is $\sqrt{\Wmat}\Tmat=\Qmat\Rmat$, and set $\Qmat_{w}=\Wmat^{-1/2}\Qmat$.
    \item Compute the kernel matrix $\Kmat=\Pmat_{w}\Emat\Pmat_{w}^t$ efficiently, where $\Pmat_{w}=\Imat_{n_{unique}}-\Qmat_w\Qmat_w^t\Wmat$.
    \item Compute a truncated eigendecomposition of $\sqrt{\Wmat}\Kmat\sqrt{\Wmat}\approx \Umat_{k-M}\Lambdamat_{k-M}\Umat_{k-M}^t$.
    \item Compute the transition matrix $\Amat_k=\begin{pmatrix}
   \Rmat^{-1} & -\Rmat^{-1}\Qmat_w^t \Wmat\Emat\Pmat_w^t\sqrt{\Wmat}\Umat_{k-M}\Lambdamat_{k-M}^{-1}\\ \nullmat& \Pmat_w^t\sqrt{\Wmat}\Umat_{k-M}\Lambdamat_{k-M}^{-1}
\end{pmatrix}$.
\item Return the design matrix $\Phimat_k=(\Qmat_w,\Wmat^{-1/2}\Umat_{k-M})$, the penalty matrix $\Gammamat_k=\text{bdiag}(\nullvec_M,\Lambdamat_{k-M}^{-1})$, the transition matrix $\Amat_k$ as well as $\text{tr}(\Lambdamat_{k-M})/\text{tr}(\sqrt{\Wmat}\Kmat\sqrt{\Wmat})$.
\end{enumerate}
\end{algorithm}
\end{document}